\documentclass[11pt]{article}

\usepackage{fullpage}
\usepackage{amsmath}
\usepackage{amsthm}
\usepackage{amssymb}
\usepackage{amsfonts}
\usepackage{hyperref}
\usepackage{color}
\usepackage{xcolor}
\usepackage{mathbbol}
\usepackage{algpseudocode}
\usepackage{multicol}
\usepackage{relsize}

\newcommand{\ignore}[1]{}

\newtheorem{theorem}{Theorem}

\newtheorem{lemma}[theorem]{Lemma}

\newtheorem{claim}{Claim}

\renewcommand{\Pr}{{\bf Pr}}
\newcommand{\E}{{\bf E}}

\newcommand{\dist}{{\rm dist}}

\DeclareMathAlphabet\mathbfcal{OMS}{cmsy}{b}{n}

\newcommand{\accepts}{\mbox{\ accepts \ }}

\newcommand{\TT}{{\boldsymbol T}}
\newcommand{\WW}{{\boldsymbol W}}
\newcommand{\ZZ}{{\boldsymbol Z}}
\newcommand{\zz}{{\boldsymbol z}}
\newcommand{\YY}{{\boldsymbol Y}}
\newcommand{\DD}{{\mathbfcal D}}
\newcommand{\yy}{{\boldsymbol y}}
\newcommand{\bb}{{\boldsymbol b}}
\newcommand{\gd}{{\boldsymbol g}}
\renewcommand{\gg}{{\boldsymbol g}}
\newcommand{\ff}{{\boldsymbol f}}
\newcommand{\hh}{{\boldsymbol h}}
\renewcommand{\SS}{{\boldsymbol S}}
\newcommand{\JJ}{{\boldsymbol J}}
\newcommand{\ggamma}{{\boldsymbol \gamma}}

\newcommand{\bone}{{\boldsymbol 1}}
\newcommand{\aalpha}{{\boldsymbol \alpha}}
\newcommand{\bbeta}{{\boldsymbol \beta}}
\newcommand{\oomega}{{\boldsymbol \omega}}

\renewcommand{\dist}{{\rm dist}}

\newcommand{\D}{{\cal D}}
\newcommand{\YES}{{\cal YES}}
\newcommand{\NO}{{\cal NO}}
\newcommand{\JUNTA}{{\cal JUNTA}}

\begin{document}

\title{Lecture Note on LCSSX's Lower Bounds for \\ Non-Adaptive Distribution-free Property Testing}
\author{{\bf Nader H. Bshouty}\\ Dept. of Computer Science\\ Technion,  Haifa, 32000\\
}

\maketitle
\begin{abstract}
In this lecture note we give Liu-Chen-Servedio-Sheng-Xie's (LCSSX) lower bound for property testing in the non-adaptive distribution-free model~\cite{LiuCSSX18}.
\end{abstract}

\section{Inroduction}
Here we give the following LCSSX's lower bound (Zhengyang Liu, Xi~Chen, Rocco~A. Servedio, Ying Sheng, and Jinyu Xie.
Distribution-free junta testing.)
\begin{theorem}\cite{LiuCSSX18} Let $k\ge 10$. Let $C$ be a class of boolean functions $f:\{0,1\}^n\to \{0,1\}$ that contains all the $k$-junta functions where $n\ge 15+2\log\log |C|$. Any non-adaptive algorithm that distribution-free $(1/3)$-tests $C$ must have query complexity at least
$$q= \frac{1}{8(1+2\lambda)^{k/2}}\cdot {2^{k/2}},$$ where
$$\lambda=\sqrt{\frac{5+\ln\ln |C|+k/2}{n}}.$$

In particular, when $\log\log|C|=o(n)$ then\footnote{Because $C$ contains all the $k$-junta functions, if $\log\log |C|=o(n)$ then $k=o(n)$} $$q={(2-o_n(1))^{k/2}}.$$
\end{theorem}

The proof in this note is the same as of LCSSX~\cite{LiuCSSX18}.

For the definition of the model and other definitions, read from~\cite{LiuCSSX18} Subsection ``{\bf Distribution-free property testing}'' in Section 1 and Section 2
and Subsection ``{\bf Junta and literals}'' in Section 2.
For other results when $C$ is the set of $k$-juntas read the introduction in~\cite{Bshouty19}.

\section{Notations}
We follow the same notations as in~\cite{LiuCSSX18}. Denote $[n]=\{1,2,\ldots,n\}$. For $X\subset [n]$ we denote by $\{0,1\}^X$
the set of all binary strings of
length $|X|$ with coordinates indexed by $i\in X$. For $x\in \{0,1\}^n$ and $X\subseteq [n]$ we write $x_X\in\{0,1\}^{X}$ to denote the projection of $x$ over coordinates in $X$.

Given a sequence $Y=(y^{(i)}:i\in [q])$ of $q$ strings in $\{0,1\}^n$ and a Boolean function $\phi:\{0,1\}^n\to \{0,1\}$, we write $\phi(Y)$ to denote the $q$-bit string $\alpha$ with $\alpha_i=\phi(y^{(i)})$ for $i\in [q]$. For a distribution~$\D$, we write $\yy\gets \D$ to denote that $\yy$ is a draw from the distribution $\D$ and $\YY=(\yy^{(i)}:i\in [q])\gets \D^q$ to denote a sequence of $q$ independent draws from the same probability distribution $\D$.

For convenience, we refer to an algorithm as a $q$-query algorithm if it makes $q$ sample queries and
$q$ black-box queries each. Such algorithms are clearly at least as powerful as those that make $q$
queries in total.
\section{Preliminary Results}
In this section we give some preliminary results

\subsection{Chernoff Bound}
We will use the following version of Chernoff Bound
\begin{lemma}\label{Chernoff}{\bf Chernoff's Bound}. Let $X_1,\ldots, X_m$ be independent random variables taking values in $\{0, 1\}$. Let $X=\sum_{i=1}^mX_i$ denotes their sum and let $\mu = \E[X]$ denotes the sum's expected value. Then
\begin{eqnarray}\Pr[X>(1+\eta)\mu]\le  \begin{cases} e^{-\frac{\eta^2\mu}{3}} &\mbox{if\ } 0< \eta\le 1 \\
e^{-\frac{\eta \mu}{3}} & \mbox{if\ } \eta>1 \end{cases} .\label{Chernoff1}
\end{eqnarray}
For $0\le \eta\le 1$ we have
\begin{eqnarray}
\Pr[X<(1-\eta)\mu]\le  e^{-\frac{\eta^2\mu}{2}}.\label{Chernoff2}
\end{eqnarray}
\end{lemma}

\subsection{Some Results in Probability}
Let $D$ be a probability distribution over a finite set $\Omega$. We will use the following (In the following two lemmas $\Pr=\Pr_D$)
\begin{lemma}\label{l0}
Let $A,B\subseteq \Omega$ where $B\not=\O$. Then $$\Pr[A|B]-\Pr[\overline{B}]\le \Pr[A]\le \Pr[A|B]+\Pr[\overline{B}].$$
\end{lemma}
\begin{proof} We have
\begin{eqnarray*}
\Pr[A]&=& \Pr[A|B]\Pr[B]+\Pr[A|\overline{B}]\Pr[\overline{B}]\\
&\le& \Pr[A|B]+\Pr[\overline{B}]
\end{eqnarray*}
and
\begin{eqnarray*}
\Pr[A]&=& 1-\Pr[\overline{A}]\\
&\ge& 1-\Pr[\overline{A}|B]-\Pr[\overline{B}]=\Pr[{A}|B]-\Pr[\overline{B}].
\end{eqnarray*}
\end{proof}

\begin{lemma}\label{l1}
Let $A,B,W\subseteq \Omega$ where $W\not=\O$. If $\Pr[A|W]\le \Pr[B|W]$ then $\Pr[A]\le \Pr[B]+\Pr[\overline{W}].$
\end{lemma}
\begin{proof} We have
\begin{eqnarray*}
\Pr[A]&=& \Pr[A|W]\Pr[W]+\Pr[A|\overline{W}]\Pr[\overline{W}]\\
&\le& \Pr[B|W]\Pr[W]+\Pr[\overline{W}]\le \Pr[B]+\Pr[\overline{W}].
\end{eqnarray*}
\end{proof}

\begin{lemma}\label{l2} {\bf Birthday Paradox}: Let $X$ be a finite set and let $\YY$ be a set obtained by making $r$ draws from $X$ uniformly at random with replacement. Then
$$\Pr[|\YY|\not=r]\le \frac{r^2}{2|X|}.$$
\end{lemma}
\begin{proof} Since for $x_1,\ldots,x_{j}\in [0,1]$, $(1-x_1)\cdots(1-x_{j})\ge 1-(x_1+\cdots+x_{j})$, we have
$$\Pr[|\YY|\not=r]=1-\prod_{i=1}^{r-1} \left(1-\frac{i}{|X|}\right)\le \frac{r^2}{2|X|}.$$
\end{proof}

\subsection{Total Variation Distance}
Let $D_1$ and $D_2$ be two probability distributions over a finite set $\Omega$. The {\it total variation distance}
between $D_1$ and $D_2$ (also called statistical distance) is
$$\|D_1-D_2\|_{tv}:=\frac{1}{2}\sum_{\omega\in \Omega}|\underset{D_1}{\Pr}[\omega]-\underset{D_2}{\Pr}[\omega]|.$$
The following lemmas are well known and easy to prove
\begin{lemma}\label{L2} The total variation distance between $D_1$ and $D_2$ is
$$\|D_1-D_2\|_{tv}=\max_{E\subseteq  \Omega}|\underset{D_1}{\Pr}[E]-\underset{D_2}{\Pr}[E]|.$$
\end{lemma}

\begin{lemma}\label{TVDX}\label{L3} Let $X: \Omega\to [0,1]$ be a random variable. Then
$$\left|\underset{D_1}{\E}[X]-\underset{D_2}{\E}[X]\right|\le \|D_1-D_2\|_{tv}.$$
\end{lemma}

\begin{lemma}\label{L4E} Let $W$ be an event such that $\underset{D_1}{\Pr}[\omega]=\underset{D_2}{\Pr}[\omega|W]$ for all $\omega\in  \Omega$. Then
$$\|D_1-D_2\|_{tv}= \underset{D_2}{\Pr}[\overline{W}].$$
\end{lemma}
\begin{proof} First, we have $\Pr_{D_1}[W]={\Pr}_{D_2}[W|W]=1$. Now
\begin{eqnarray*}
\|D_1-D_2\|_{tv}=\max_{E\subseteq \Omega}|\underset{D_1}{\Pr}[E]-\underset{D_2}{\Pr}[E]|
\overset{E=W}{\ge} \underset{D_2}{\Pr}[\overline{W}]
\end{eqnarray*}
and by Lemma~\ref{l0}, for any $E$,
$$|\underset{D_1}{\Pr}[E]-\underset{D_2}{\Pr}[E]|=|\underset{D_2}{\Pr}[E|W]-\underset{D_2}{\Pr}[E]|\le \underset{D_2}{\Pr}[\overline{W}].$$
\end{proof}

\begin{lemma}\label{L4} Let $W$ be an event such that $\underset{D_1}{\Pr}[\omega|W]=\underset{D_2}{\Pr}[\omega|W]$ for all $\omega\in  \Omega$. Then
$$\|D_1-D_2\|_{tv}\le \underset{D_1}{\Pr}[\overline{W}]+\underset{D_2}{\Pr}[\overline{W}].$$
\end{lemma}
\begin{proof} Let $D_3$ be the conditional distribution of $D_1$ given $W$. Then $\Pr_{D_2}[\omega|W]=\Pr_{D_1}[\omega|W]=\Pr_{D_3}[\omega]$. By Lemma~\ref{L4E}, $\|D_1-D_3\|_{tv}= \Pr_{D_1}[\overline{W}]$ and $\|D_2-D_3\|_{tv}=\Pr_{D_2}[\overline{W}]$ and therefore
$$\|D_1-D_2\|_{tv}\le \|D_1-D_3\|_{tv}+\|D_2-D_3\|_{tv}=\underset{D_1}{\Pr}[\overline{W}]+\underset{D_2}{\Pr}[\overline{W}].$$
\end{proof}

\begin{lemma}\label{L5} Let $D_1$ and $D_2$ be two probability distributions over $\Omega_1\times \Omega_2$.
If for every $\omega\in \Omega_1$, $\Pr_{(\oomega_1,\oomega_2)\gets D_1}[\oomega_1=\omega]=\Pr_{(\oomega_1,\oomega_2)\gets D_2}[\oomega_1=\omega]$ then the total variation distance between the distributions $D_1$ and $D_2$ is less than or equal to the maximum over $w_1\in \Omega_1$ of the total variation distance between the distributions of $\oomega_2$ conditioning on $\oomega_1=\omega_1$ in $D_1$ and $D_2$.
\end{lemma}

\subsection{Lower Bound Technique}
Our goal is to show that there exists no $q$-query non-adaptive (randomized)
algorithm that distribution-free $(1/3)$-tests $C$.

We can think of a distribution-free $(1/3)$-tester for $C$ as a randomized algorithm $T$ that receives as an input a pair $(\phi,\D)$ where $\phi : \{0, 1\}^n \to \{0, 1\}$ and $\D$ is a probability distribution over $\{0,1\}^n$. If $\phi\in C$ then $T$ accepts with probability at least $2/3$ and if $f$ is $(1/3)$-far from every function in $C$ with respect to $\D$ then it rejects with probability at least $2/3$.

The (folklore) technique introduced here shows that it is enough to focus on $q$-query non-adaptive {\bf deterministic} algorithms. Such
an algorithm $A$ consists of two deterministic maps $A_1$ and $A_2$ works as
follows. Upon an input pair $(\phi,\D)$, where $\phi : \{0, 1\}^n \to \{0, 1\}$ and $\D$ is a probability distribution
over $\{0, 1\}^n$, the algorithm receives in the first phase a sequence $Y = (y^{(i)} : i \in [q])$ of $q$ strings
(which should be thought of as samples from $\D$) and a binary string $\alpha = \phi(Y )$ of length~$q$. In the
second phase, the algorithm $A$ uses the first map $A_1$ to obtain a sequence of $q$ strings $Z = (z^{(i)} : i \in
[q]) = A_1 (Y, \alpha)$ and feeds them to the black-box oracle. Once the query results $\beta = \phi(Z)$ are back,
$A_2 (Y, \alpha, \beta)$ returns either $0$ or $1$ in which cases the algorithm $A$ either rejects or accepts, respectively. Notice that we do not need to include $Z$ as an input of $A_2$, since it
is determined by $Y$ and $\alpha$. A
randomized algorithm $T$ works similarly and consists of two similar maps $T_1$ and $T_2$ but both are
randomized. The following are the two algorithms $A$ and $T$. The (infinite length) strings $s_1$ and $s_2$ are two random seeds
\\

\noindent\fbox{%
    \parbox{\textwidth}{%
\begin{multicols}{2}
  \begin{minipage}{0.45\textwidth}
    {\bf Deterministic Algorithm $A$}
    \begin{enumerate}
    \setlength\itemsep{.1em}
    \item Input $(\phi,\D)$
    \item Get $Y= (y^{(i)} : i \in [q])$
    \item $\alpha=\phi(Y)$
    \item $Z=(z^{(i)}:i\in [q])=A_1(Y,\alpha)$
    \item $\beta=\phi(Z)$
    \item Output $A_2(Y,\alpha,\beta)$
    \end{enumerate}
  \end{minipage}
\begin{minipage}{0.45\textwidth}
    {\bf Randomized Algorithm $T$}
    \begin{enumerate}
   \setlength\itemsep{.1em}
    \item Input $(\phi,\D)$
    \item Get $Y= (y^{(i)} : i \in [q])$
    \item $\alpha=\phi(Y)$
    \item $Z=(z^{(i)}:i\in [q])=T_1(Y,\alpha,s_1)$
    \item $\beta=\phi(Z)$
    \item Output $T_2(Y,\alpha,\beta,s_2)$
    \end{enumerate}
  \end{minipage}
\end{multicols}
}}
\\

Given the above deterministic algorithm, unlike typical deterministic algorithms, whether $A$ accepts or not
depends on not only $(\phi,\D)$ but also the sample strings $\YY\gets \D^q$ it draws. Formally, we have
\begin{eqnarray*}\Pr[A \mbox{\ accepts\ }(\phi,\D)]&=&\underset{\YY\gets \D^q}{\Pr}[A \mbox{\ accepts\ }(\phi,\D)]\\
&=& \underset{\YY\gets \D^q}{\Pr}[A_2(\YY,\phi(\YY),\phi(A_1(\YY,\phi(\YY))))=1].
\end{eqnarray*}
For the randomized algorithm $T$ we have
\begin{eqnarray*}\Pr[T \mbox{\ accepts\ }(\phi,\D)]&=& \underset{s_1,s_2,\YY\gets \D^q}{\Pr}[T \mbox{\ accepts\ }(\phi,\D)]\\&=&
\underset{s_1,s_2,\YY\gets \D^q}{\Pr}[T_2(\YY,\phi(\YY),\phi(T_1(\YY,\phi(\YY),s_1)),s_2)=1].\end{eqnarray*}

We now prove
\begin{lemma} \cite{LiuCSSX18}
Let $\YES$ and $\NO$ be probability
distributions over pairs $(\phi,\D)$, where $\phi:\{0,1\}^n\to \{0,1\}$ is a Boolean function over $n$ variables and
$\D$ is a distribution over $\{0, 1\}^n$. For clarity, we use $( f ,\D)$ to denote pairs in the support of $\YES$ and
$(g,\D)$ to denote pairs in the support of $\NO$. Suppose $\YES$ and $\NO$ satisfy
\begin{description}
\item[C1:] Every $( f ,\D)$ in the support of $\YES$
satisfies that $f$ is in $C$.
\item[C2:] With probability at least $12/13$, $(\gd,\DD) \gets \NO$ satisfies
that $\gd$ is $(1/3)$-far from every function in $C$ with respect to $\DD$.
\item[C3:] Any $q$-query
non-adaptive {\it deterministic} algorithm must behave similarly when it is run on $(\ff,\DD) \gets\YES$
versus $(\gd,\DD) \gets \NO$: That is,
any $q$-query deterministic algorithm $A$ satisfies
$$\left| \underset{(\ff,\DD)\gets \YES}{\E}[\Pr[A\mbox{\ accepts\ }(\ff,\DD)]]-
\underset{(\gd,\DD)\gets \NO}{\E}[\Pr[A\mbox{\ accepts\ }(\gd,\DD)]]\right|\le \frac{1}{4}.$$
\end{description}
Then any non-adaptive (randomized) algorithm $T$ that distribution-free $(1/3)$-tests $C$ must have query complexity
at least $q$.
\end{lemma}
\begin{proof} Assume for a contradiction that
there exists a $q$-query non-adaptive randomized algorithm~$T_{s_1,s_2}$ that distribution-free $(1/3)$-tests
$C$ where $s_1$ and $s_2$ are the random seeds of the algorithm. Then, by {\bf C1}, for every $(f,\D)$ in the support of $\YES$ we have
$\Pr[T_{s_1,s_2}\mbox{\ accepts\ }(f,\D)]\ge 2/3$. Therefore,
\begin{eqnarray}
\underset{(\ff,\DD)\gets \YES}{\E}[\Pr [T_{s_1,s_2} \mbox{\ accepts\ }(\ff,\DD)]]\ge \frac{2}{3}.\label{pq1}
\end{eqnarray}
Define
$U:=[\gd$ is $(1/3)$-far from every function in $C$ with respect to $\DD]$ and $\WW:=\Pr[T_{s_1,s_2}\mbox{\ accepts\ }(\gg,\DD)]$. Then, by {\bf C2},
\begin{eqnarray}
\underset{(\gd,\DD)\gets \NO}{\E}[\WW]&=&\underset{(\gd,\DD)\gets \NO}{\E}[\WW|U]\underset{(\gd,\DD)\gets \NO}{\Pr}[U]+\underset{(\gd,\DD)\gets \NO}{\E}[\WW|\overline{U}]\underset{(\gd,\DD)\gets \NO}{\Pr}[\overline{U}]\nonumber\\
&\le&\underset{(\gd,\DD)\gets \NO}{\E}[\WW|U]+\underset{(\gd,\DD)\gets \NO}{\Pr}[\overline{U}]\nonumber\\
&\le& \frac{1}{3}+\frac{1}{13}<\frac{5}{12}.\label{pq2}
\end{eqnarray}
By (\ref{pq1}) and (\ref{pq2}) we have that
$$ \underset{(\ff,\DD)\gets \YES}{\E}[\Pr[T_{s_1,s_2}\mbox{\ accepts\ }(\ff,\DD)]]-
\underset{(\gd,\DD)\gets \NO}{\E}[\Pr[T_{s_1,s_2}\mbox{\ accepts\ }(\gd,\DD)]]> \frac{2}{3}-\frac{5}{12}=\frac{1}{4}.$$
Since,
\begin{eqnarray*}\Pr[T_{s_1,s_2} \mbox{\ accepts\ }(\phi,\D)]&=& \underset{s_1,s_2,\YY\gets \D^q}{\Pr}[T_{s_1,s_2} \mbox{\ accepts\ }(\phi,\D)]\\&=&
\underset{s_1,s_2}{\E}\left[\underset{\YY\gets \D^q}{\Pr}[T_{s_1,s_2} \mbox{\ accepts\ }(\phi,\D)]\right]
\end{eqnarray*}
we have
$$ \underset{s_1,s_2}{\E}\left[\underset{(\ff,\DD)\gets \YES}{\E}\left[\underset{\YY\gets\DD^q}{\Pr}[T_{s_1,s_2}\mbox{\ accepts\ }(\ff,\DD)]\right]-
\underset{(\gd,\DD)\gets \NO}{\E}\left[\underset{\YY\gets\DD^q}{\Pr}\left[T_{s_1,s_2}\mbox{\ accepts\ }(\gd,\DD)\right]\right]\right]> \frac{1}{4}.$$
Thus, there exist $s_1'$ and $s_2'$, and therefore a $q$-query nonadaptive
deterministic algorithm $A=T_{s_1',s_2'}$, that satisfies
$$\left| \underset{(\ff,\DD)\gets \YES}{\E}[\Pr[A\mbox{\ accepts\ }(\ff,\DD)]]-
\underset{(\gd,\DD)\gets \NO}{\E}[\Pr[A\mbox{\ accepts\ }(\gd,\DD)]]\right|> \frac{1}{4}.$$
A contradiction to {\bf C3}.
\end{proof}

\section{The $\YES$ and $\NO$ Distributions}
Given $J \subseteq [n]$, we partition $\{0, 1\}^n$ into sections (with respect to $J$) where the $z$-{\it section}, $z \in \{0, 1\}^J$,
consists of those $x \in \{0, 1\}^n$ that have $x_J = z$. We write $\JUNTA
_J$ to denote the uniform distribution
over all juntas over $J$ . More precisely, a Boolean function $\hh : \{0, 1\}^n \to \{0, 1\}$ drawn
from $\JUNTA_J$ is generated as follows: For each $z \in \{0, 1\}^J$, a bit $\bb(z)$ is chosen independently
and uniformly at random, and for each $x \in \{0, 1\}^n$ the value of $\hh(x)$ is set to $\bb(x_J )$. That is, if $x$ is in the $z$-section then $f(x)=\bb(z)$.

We now define two probability distributions:
Let $$m={18\ln|C|}.$$

\noindent\fbox{%
    \parbox{\textwidth}{%
\noindent
\centerline{\underline{\bf The probability distribution $\YES$}}

A pair $(\ff,\DD)$ drawn from $\YES$ is generated as follows:
\begin{enumerate}
\item Draw a subset $\JJ$ of $[n]$ of size $k$ uniformly at random
\item Draw a subset
$\SS$ of $\{0, 1\}^n$ of size $m$ uniformly at random.
\item Draw $\ff \gets \JUNTA_\JJ$
\item Set $\DD$ to be the uniform distribution over $\SS$.
\end{enumerate}
}}\\ \\

\noindent
\fbox{%
    \parbox{\textwidth}{%
\noindent
\centerline{\underline{\bf The probability distribution $\NO$}}

A pair $(\gd,\DD)$ drawn from $\NO$ is generated as follows:
\begin{enumerate}
\item Draw a subset $\JJ$ of $[n]$ of size $k$ uniformly at random
\item Draw a subset $\SS$ of $\{0, 1\}^n$ of size $m$ uniformly at random.
\item Draw $\hh\gets  \JUNTA_\JJ$. We usually refer to
$\hh$ as the “background junta.”
\item Draw a map $\ggamma : \SS \to \{0, 1\}$ uniformly at random by choosing
a bit independently and uniformly at random for each string in $\SS$.
\item The distribution $\DD$ is set to be the uniform distribution over $\SS$, which is the same as $\YES$.
\item The function $\gd : \{0, 1\}^n \to \{0, 1\}$ is defined using $\hh, \SS$ and $\ggamma$ as follows:
$$\gd(x)=\left\{
\begin{array}{ll}
\ggamma(x)&  x\in \SS \\
\hh(x)& x\not\in \SS, (\forall y\in \SS)\ x_\JJ\not=y_\JJ \mbox{\ or\ } d(x,y)>(0.5-\lambda)n\\
\ggamma(y)& x\not\in \SS, (\exists y\in \SS)\ x_\JJ=y_\JJ \mbox{\ and\ } d(x,y)\le (0.5-\lambda)n\ \ \  (*)
\end{array}
\right.
$$
{(*) The choice of the tie-breaking
rule here is not important; we can, for example, order the elements of $S$ in a lexicographic order $(s^{(i)}:i\in [m])$ and define $g(x)=\ggamma(s^{(i)})$ for the smallest $i$ that satisfies $x_\JJ=s^{(i)}_\JJ \mbox{\ and\ } d(x,s^{(i)})\le (0.5-\lambda)n~$. This makes $\gd$ well defined.}
\end{enumerate}
}} \\

For technical reasons that will become clear in the sequel we use $\YES^*$
to denote the probability
distribution supported over triples $( f ,\D, J)$, with $(\ff,\DD, \JJ) \gets \YES^*$
being generated by the
same steps above. So, the only difference is that we include $\JJ$ in elements of $\YES^*$. Similarly, we let $\NO^*$
denote the distribution supported on triples $(g,\D, J)$ as generated above.

To understand the intuition behind the above definitions, read subsubsection {\bf The lower bound} in subsection 1.2 and the last paragraph in page 1:17 in~\cite{LiuCSSX18} (when $C$ is the class of all $k$-juntas).

\section{The Proofs of C1 and C2}
In this section we prove
\begin{description}
\item[C1:] Every $( f ,\D)$ in the support of $\YES$
satisfies that $f$ is in $C$.
\item[C2:] With probability at least $12/13$, $(\gd,\DD) \gets \NO$ satisfies
that $\gd$ is $(1/3)$-far from every function in $C$ with respect to $\DD$.
\end{description}

\noindent {\bf Proof of C1}: By the definition of $\YES$ we have that $f$ is $k$-junta. Since $C$ contains all the $k$-juntas we have that $f$ is in $C$. \qed

\noindent {\bf Proof of C2}: Let $\beta\in C$. Since $\DD$ is the uniform distribution over $\SS$, we have that $dist_\DD(\gd,\beta)$ is equal to the fraction of strings $z \in \SS$ such that $\ggamma (z) \not= \beta(z)$.
By the union bound, we have
\begin{eqnarray}
\underset{(\gd,\DD)\gets \NO}{\Pr}\left[\dist(\gd,C)<1/3\right]&=&
\underset{(\gd,\DD)\gets \NO}{\Pr}\left[(\exists \beta\in C)\ \dist(\gd,\beta)<1/3\right]\nonumber\\
&=&
\underset{(\gd,\DD)\gets \NO}{\Pr}\left[(\exists \beta\in C)\ \underset{\zz\gets \DD}{\Pr}[\gd(\zz)\not=\beta(\zz)]<1/3\right]\nonumber\\
&\le&
|C|\cdot\max_{\beta\in C}\underset{(\gd,\DD)\gets \NO}{\Pr}\left[\underset{\zz\gets \DD}{\Pr}[\gd(\zz)\not=\beta(\zz)]<1/3\right].\label{mm1}
\end{eqnarray}

Now let $\bone_{\gd\not=\beta}(\zz)$ be the indicator random variable of $\gd(\zz)\not=\beta(\zz)$, i.e, $\bone_{\gd\not=\beta}(\zz)=1$ if $\gd(\zz)\not=\beta(\zz)$ and zero otherwise. Since each bit $\ggamma (\zz)$, $\zz\in \SS$, is drawn independently and uniformly at random, we have that, for every $\zz\in \SS$,
$$\underset{(\gd,\DD)\gets \NO}{\E}[\bone_{\ggamma\not=\beta}(\zz)]=\frac{1}{2}.$$
Then, by Chernoff bound (\ref{Chernoff2}) in Lemma~\ref{Chernoff} ($m= 18\ln |C|$, $k\ge 10$, $C$ contains all $k$-Junta functions and therefore $|C|\ge 2^{2^k}>13$),
\begin{eqnarray*}
\underset{(\gd,\DD)\gets \NO}{\Pr}\left[ \underset{\zz\gets\DD}{\Pr}[\gd(\zz)\not=\beta(\zz)]<\frac{1}{3} \right]&=&
\underset{(\gd,\DD)\gets \NO}{\Pr}\left[ \sum_{\zz\in \SS}\bone_{\gd\not=\beta}(\zz)<\frac{1}{3} m\right]\\
&\le& e^{-m/9}\le \frac{1}{13|C|}.
\end{eqnarray*}
Therefore
\begin{eqnarray}
|C|\cdot\max_{\beta\in C}\underset{(\gd,\DD)\gets \NO}{\Pr}\left[ \underset{\zz\gets \DD}{\Pr}[\gd(\zz)\not=\beta(\zz)]<\frac{1}{3} \right]\le \frac{1}{13}.\label{mm3}
\end{eqnarray}
By (\ref{mm1}) and (\ref{mm3}) we get
\begin{eqnarray*}
\underset{(\gd,\DD)\gets \NO}{\Pr}\left[\dist(\gd,C)<1/3\right]\le \frac{1}{13}.\qed
\end{eqnarray*}

\section{The Proof of C3}
In this section we prove
\begin{description}
\item[C3:] Any $q$-query
non-adaptive {\it deterministic} algorithm must behave similarly when it is run on $(\ff,\DD) \gets\YES$
versus $(\gd,\DD) \gets \NO$: That is,
any $q$-query deterministic algorithm $A$ satisfies
$$\left| \underset{(\ff,\DD)\gets \YES}{\E}[\Pr[A\mbox{\ accepts\ }(\ff,\DD)]]-
\underset{(\gd,\DD)\gets \NO}{\E}[\Pr[A\mbox{\ accepts\ }(\gd,\DD)]]\right|\le \frac{1}{4}.$$
\end{description}

Let $A$ be a $q$-query non-adaptive deterministic algorithm where
$$q= \frac{1}{8(1+2\lambda)^{k/2}}\cdot {2^{k/2}},$$ and
$$\lambda=\sqrt{\frac{5+\ln\ln |C|+k/2}{n}}.$$
We will use the following definition. Let $Y = (y_i : i \in [q])$ be a sequence of $q$ strings in $\{0, 1\}^n$, $\alpha$ be a $q$-bit string, and $J \subset [n]$ be a set of size $k$. We say that $(Y, \alpha, J )$ is consistent if
$$\alpha_i = \alpha_j \mbox{\ for all\ } i, j \in [q]\mbox{\ with\ }y^{(i)}_J = y^{(j)}_J.$$
Given a consistent triple $(Y, \alpha, J )$, we write $\JUNTA_{Y,\alpha, J}$ to denote the uniform distribution over all juntas $h$ over $J$ that are consistent with $(Y, \alpha)$. More precisely, a draw of $\hh \gets \JUNTA_{Y,\alpha, J}$ is generated as follows: For each $z \in\{0, 1\}^J$, if there exists a $y^{(i)}$ such that $y^{(i)}_J = z$, then $\hh(x)$ is set to $\alpha_i$ for all $x \in \{0, 1\}^n$ with $x_J = z$; if no such $y^{(i)}$ exists, then a uniform random bit
$\bb(z)$ is chosen independently and $\hh(x)$ is set to $\bb(z)$ for all $x$ with $x_J = z$.

To prove {\bf C3}, we first derive from $A$ the following randomized algorithm $A'$ that works on triples
$(\phi,D,J)$ from the support of either $\YES^*$
or $\NO^*$. Again for clarity we use $\phi$ to denote a function
from the support of $\YES/\YES^*$
or $\NO/\NO^*$, $f$ to denote a function from $\YES/\YES^*$
and $g$ to
denote a function from $\NO/\NO^*$.\\

\noindent
\fbox{%
    \parbox{\textwidth}{%
\begin{multicols}{2}
  \begin{minipage}{0.45\textwidth}
    {\bf Deterministic Algorithm $A$}
    \begin{enumerate}
    \setlength\itemsep{.1em}
    \item Input $(\phi,\D)$
    \item $\YY\gets \D^q$
    \item $\aalpha=\phi(\YY)$\\
    \item $\ZZ=A_1(\YY,\aalpha)$
    \item \ \\ $\bbeta=\phi(\ZZ)$
    \item Output $A_2(\YY,\aalpha,\bbeta)$
    \end{enumerate}
  \end{minipage}
\begin{minipage}{0.45\textwidth}
    {\bf Randomized Algorithm $A'$}
    \begin{enumerate}
   \setlength\itemsep{.1em}
    \item Input $(\phi,\D,J)$
    \item $\YY\gets \D^q$;
    \item $\aalpha=\phi(\YY)$\\ If $(\YY,\aalpha,J)$ is not consistent reject
    \item $\ZZ=A_1(\YY,\aalpha)$
    \item Draw $\hh'\gets\JUNTA_{\YY,\aalpha,J}$; \\ $\bbeta=\hh'(\ZZ)$
    \item Output $A_2(\YY,\aalpha,\bbeta)$
    \end{enumerate}
  \end{minipage}
\end{multicols}}} \\

From the description of $A'$ above, we have
$$\Pr[A'\mbox{\ accepts\ }(\phi,\D,J)]=\underset{\YY,\hh'}{\Pr}[(\YY,\aalpha,J)\mbox{\ is consistent and\ } A_2(\YY,\aalpha,\hh'(\ZZ))=1].$$

To prove {\bf C3} we will prove the following
\begin{description}
\item[C3.1] $A'$ behaves similarly on $\YES^*$ and $\NO^*$, i.e,
$$\left|\underset{(\ff,\DD,\JJ)\gets \YES^*}{\E}[\Pr[A'\accepts (\ff,\DD,\JJ)]]-\underset{(\gd,\DD,\JJ)\gets \NO^*}{\E}[\Pr[A'\accepts (\gd,\DD,\JJ)]]\right|\le\frac{1}{8}.$$
\item[C3.2] $A$ and $A'$ behave identically on $\YES$ and $\YES^*$, respectively. i.e,
$$\underset{(\ff,\DD,\JJ)\gets \YES^*}{\E}[\Pr[A'\accepts (\ff,\DD,\JJ)]]=\underset{(\ff,\DD,\JJ)\gets \YES^*}{\E}[\Pr[A\accepts (\ff,\DD)]].$$
\item[C3.3] $A'$ and $A$ behave similarly on $\NO$ and $\NO^*$, respectively. i.e,
$$\left|\underset{(\gd,\DD,\JJ)\gets \NO^*}{\E}[\Pr[A'\accepts (\gd,\DD,\JJ)]]-\underset{(\gd,\DD,\JJ)\gets \NO^*}{\E}[\Pr[A\accepts (\gd,\DD)]]\right|\le \frac{1}{8}.$$
\end{description}

\begin{table}[h]
\begin{tabular}{ccc}
\cline{1-1} \cline{3-3}
 \multicolumn{1}{|l|}{$\underset{(\ff,\DD,\JJ)\gets \YES^*}{\E}[\Pr[A'\accepts (\ff,\DD,\JJ)]]$}  & \ $\overset{{\mathlarger{\le \frac{1}{8}}}}{\rule{1.2cm}{1pt}}$ &
 \multicolumn{1}{|l|}{$\underset{(\gd,\DD,\JJ)\gets \NO^*}{\E}[\Pr[A'\accepts (\gd,\DD,\JJ)]]$}\\
 \cline{1-1} \cline{3-3}
 $\vrule\vrule\vrule$ $ \ =$ && $\vrule\vrule\vrule$ $\le\frac{1}{8}$\\
 \cline{1-1} \cline{3-3}
 \multicolumn{1}{|l|}{$\overset{\phantom{H}}{\underset{(\ff,\DD,\JJ)\gets \YES^*}{\E}}[\Pr[A\accepts (\ff,\DD)]]$ }&&
 \multicolumn{1}{|l|}{$\underset{(\gd,\DD,\JJ)\gets \NO^*}{\E}[\Pr[A\accepts (\gd,\DD)]]$ }\\
 \cline{1-1} \cline{3-3}
\end{tabular}
\end{table}

Obviously, {\bf C3.1-C3.3} imply {\bf C3}.

\subsection{Proof of C3.1}
In this subsection we prove
\begin{description}
\item[C3.1] $A'$ behaves similarly on $\YES^*$ and $\NO^*$, i.e,

$$\left|\underset{(\ff,\DD,\JJ)\gets \YES^*}{\E}[\Pr[A'\accepts (\ff,\DD,\JJ)]]-\underset{(\gd,\DD,\JJ)\gets \NO^*}{\E}[\Pr[A'\accepts (\gd,\DD,\JJ)]]\right|\le\frac{1}{8}.$$
\end{description}

\noindent
\fbox{%
    \parbox{\textwidth}{%
\begin{multicols}{2}
  \begin{minipage}{0.45\textwidth}
{\bf Algorithm $A'$ - $\YES$ distribution}
    \begin{enumerate}
   \setlength\itemsep{.1em}
    \item $(\ff,\DD,\JJ)\gets \YES^*$
    \item $\YY\gets \DD^q$;
    \item $\aalpha=\ff(\YY)$\\ If $(\YY,\aalpha,\JJ)$ is not consistent reject
    \item $\ZZ=A_1(\YY,\aalpha)$
    \item Draw $\hh'\gets\JUNTA_{\YY,\aalpha,\JJ}$; \\ $\bbeta=\hh'(\ZZ)$
    \item Output $A_2(\YY,\aalpha,\bbeta)$
    \end{enumerate}
  \end{minipage}
\begin{minipage}{0.45\textwidth}
    {\bf Algorithm $A'$ - $\NO$ distribution}
    \begin{enumerate}
   \setlength\itemsep{.1em}
    \item $(\gd,\DD,\JJ)\gets \NO^*$
    \item $\YY\gets \DD^q$;
    \item $\aalpha=\gd(\YY)$\\ If $(\YY,\aalpha,\JJ)$ is not consistent reject
    \item $\ZZ=A_1(\YY,\aalpha)$
    \item Draw $\hh'\gets\JUNTA_{\YY,\aalpha,\JJ}$; \\ $\bbeta=\hh'(\ZZ)$
    \item Output $A_2(\YY,\aalpha,\bbeta)$
    \end{enumerate}
  \end{minipage}
\end{multicols}}}\\

We say $Y$ is {\it scattered} by $J$ if there is no $i\not=j$ such that $y^{(i)}_J = y^{(j)}_J$.
The following claim shows that $\YY$ is scattered by $\JJ$ with high probability.

\begin{claim}\label{claim} We have that $\YY$ is scattered by $\JJ$ with probability at least $15/16$
\end{claim}
\begin{proof} We fix $J$ and show that $\YY$ is scattered by $J$ with probability at least $15/16$.
We now define the following distributions $D_1$ and $D_2$ for $Y$.
\begin{enumerate}
\item $D_1$: Draw a subset $\SS$ of $\{0,1\}^n$ of size $m$ uniformly at random. Then choose $q$ strings $\YY=(\yy^{(i)}:i\in[q])$ independently and uniformly at random from $\SS$ with replacement.
\item $D_2$: Choose $q$ strings $\YY=(\yy^{(i)}:i\in[q])$ independently and uniformly at random from $\{0,1\}^n$ with replacement.
\end{enumerate}
Let $F$ be the event: $\YY \mbox{\ is not scattered by}\ J$. We need to show that
$$\underset{\YY\gets D_1}{\Pr}[F]\le \frac{1}{16}.$$
Let $U$ be the event that the strings in $\YY$ are distinct. It is clear that for any event $E$ we have that $\Pr_{\YY\gets D_1}[E|U]=\Pr_{\YY\gets D_2}[E|U]$. By Lemma~\ref{L4} and Lemma~\ref{l2}, the total variation distance between $D_1$ and $D_2$ is ($q\le 2^{k/2-3}$ and $m= 18\ln |C|\ge 2^{k}$)
$$\|D_1-D_2\|_{tv}\le \underset{\YY\gets D_1}{\Pr}[\overline{U}]+\underset{\YY\gets D_2}{\Pr}[\overline{U}]\le \frac{q^2}{2m}+\frac{q^2}{2^{n+1}}\le \frac{q^2}{m}\le \frac{1}{32}.$$ Since, by Lemma~\ref{L3}, $\Pr_{\YY\gets D_1}[F]\le \Pr_{\YY\gets D_2}[F]+1/32$, it remains to show that $\Pr_{\YY\gets D_2}[F]\le 1/32$.

Since $(\yy^{(i)}:i\in[q])$  are chosen independently and uniformly at random from $\{0,1\}^n$ with replacement, we have that $(\yy^{(i)}_J:i\in[q])$ are chosen independently and uniformly at random from $\{0,1\}^k$ with replacement. Thus, by Lemma~\ref{l2} ($q\le 2^{k/2-3}$),
\begin{eqnarray}
\underset{\YY\gets D_2}{\Pr}[F]\le \frac{q^2}{2^{k+1}}\le \frac{1}{32}\label{Cla03}
\end{eqnarray}
and the result follows.
\end{proof}

Since $A'$ runs on $(Y, \alpha, J )$, by Lemma~\ref{TVDX}, it suffices to show that the distributions of
$(\YY,\aalpha, \JJ)$ induced from $\YES^*$ and $\NO^*$
have total variation distance less than or equal to $1/8$. For this purpose, we
first note that the distributions of $(\YY, \JJ)$ induced from $\YES^*$
and $\NO^*$ are identical: In both cases,
$\YY$ and $\JJ$ are independent; $\JJ$ is a random subset of $[n]$ of size $k$; $\YY$ is obtained by first sampling a subset $\SS$ of $\{0, 1\}^n$ of size $m$ and then drawing a sequence of $q$ strings from $\SS$ with replacement.

Fix any $(Y, J)$ in the support of $(\YY, \JJ)$. By Lemma~\ref{L5}, it is enough to show that the total variation of the distributions of $\aalpha$ conditioning on $(\YY, \JJ) = (Y, J)$ in the $\YES^*$
case and the $\NO^*$ case is less than $1/8$.

Fix any $(Y, J)$ in the support of $(\YY, \JJ)$ such that $Y$ is scattered by $J$. By Claim~\ref{claim} and Lemma~\ref{L4} it is enough to show that the distributions of $\aalpha$ conditioning on $(\YY, \JJ) = (Y, J)$ in the $\YES^*$
case and the $\NO^*$ case are identical.

For $Y=(y^{(i)}: i\in [q])$ the string $\aalpha=(\aalpha_i:i\in [q])$ is uniform over strings of length $q$
in both cases. This is trivial for $\NO^*$.
For $\YES^*$ note that $\aalpha$ is determined by the random $k$-junta
$\ff \gets \JUNTA_J$ ; the claim follows from the assumption that $Y$ is scattered by $J$.

\subsection{Proof of C3.2}
In this subsection we prove
\begin{description}
\item[C3.2] $A$ and $A'$ behave identically on $\YES$ and $\YES^*$, respectively. i.e,
$$\underset{(\ff,\DD,\JJ)\gets \YES^*}{\E}[\Pr[A'\accepts (\ff,\DD,\JJ)]]=\underset{(\ff,\DD,\JJ)\gets \YES^*}{\E}[\Pr[A\accepts (\ff,\DD)]].$$

\end{description}

\noindent
\fbox{%
    \parbox{\textwidth}{%
\begin{multicols}{2}
\begin{minipage}{0.45\textwidth}
    {\bf Algorithm $A'$}
    \begin{enumerate}
   \setlength\itemsep{.1em}
    \item $(\ff,\DD,\JJ)\gets \YES^*$
    \item $\YY\gets \DD^q$
    \item $\aalpha=\ff(\YY)$\\ If $(\YY,\aalpha,\JJ)$ is not consistent reject
    \item $\ZZ=A_1(\YY,\aalpha)$
    \item Draw $\hh'\gets\JUNTA_{\YY,\aalpha,\JJ}$; \\ $\bbeta=\hh'(\ZZ)$
    \item Output $A_2(\YY,\aalpha,\bbeta)$
    \end{enumerate}
  \end{minipage}
  \begin{minipage}{0.45\textwidth}
    {\bf Algorithm $A$}
    \begin{enumerate}
    \setlength\itemsep{.1em}
    \item $(\ff,\DD,\JJ)\gets \YES^*$
    \item $\YY\gets \DD^q$
    \item $\aalpha=\ff(\YY)$\\
    \item $\ZZ=A_1(\YY,\aalpha)$
    \item Let \\ $\bbeta=\ff(\ZZ)$
    \item Output $A_2(\YY,\aalpha,\bbeta)$
    \end{enumerate}
  \end{minipage}
\end{multicols}}}\\

For the first expectation in {\bf C3.2}, since the triple $(\YY,\aalpha,\JJ)$ on
which we run $A'$ is always consistent, we can rewrite it as the probability that
$$A_2(\YY,\aalpha,\hh'(A_1(\YY,\aalpha)))=1,$$
where $(\ff,\DD,\JJ)\gets \YES^*$, $\YY\gets \DD^q$, $\aalpha=\ff(\YY)$ and $\hh'\gets \JUNTA_{\YY,\aalpha,\JJ}$.

The second expectation is equal to the probability that
$$A_2(\YY,\aalpha,\ff(A_1(\YY,\aalpha)))=1$$ where $(\ff,\DD,\JJ)\gets \YES^*$, $\YY\gets \DD^q$ and $\aalpha=\ff(\YY)$.

To show that these two probabilities are equal, we first note that the distributions of $(\YY,\aalpha, \JJ)$ are
identical. Fixing any triple $(Y, \alpha, J )$ in the support of $(\YY,\aalpha, \JJ)$, which must be consistent, we claim that the distribution of $\ff$ conditioning on $(\YY,\aalpha, \JJ)=(Y, \alpha, J )$  is exactly $\JUNTA_{\YY,\aalpha, \JJ}$. This is because, for each $z \in \{0, 1\}^J$ , if $y^{(i)}_J = z$ for some $y^{(i)}$ in $Y$, then we have $\ff (x) = \alpha_i$ for all strings $x$
with $x_J = z$; otherwise, we have $\ff (x) = \bb(z)$ for all $x$ with $x_J = z$, where $\bb(z)$ is an independent and uniform bit. This is the same as how $\hh' \gets \JUNTA_{Y,\alpha, J}$ is generated. It follows directly from this claim that the two probabilities are the same. This finishes the proof of {\bf C3.2}.

\subsection{Proof of C3.3}
In this subsection we prove
\begin{description}
\item[C3.3] $A'$ and $A$ behave similarly on $\NO$ and $\NO^*$, respectively. i.e,
$$\left|\underset{(\gd,\DD,\JJ)\gets \NO^*}{\E}[\Pr[A'\accepts (\gd,\DD,\JJ)]]-\underset{(\gd,\DD,\JJ)\gets \NO^*}{\E}[\Pr[A\accepts (\gd,\DD)]]\right|\le \frac{1}{8}.$$
\end{description}

\noindent
\fbox{%
    \parbox{\textwidth}{%
\begin{multicols}{2}
\begin{minipage}{0.45\textwidth}
    {\bf Algorithm $A'$}
    \begin{enumerate}
   \setlength\itemsep{.1em}
    \item $(\gg,\DD,\JJ)\gets \NO^*$
    \item $\YY\gets \DD^q$
    \item $\aalpha=\gg(\YY)$\\ If $(\YY,\aalpha,\JJ)$ is not consistent reject
    \item $\ZZ=A_1(\YY,\aalpha)$
    \item Draw $\hh'\gets\JUNTA_{\YY,\aalpha,\JJ}$; \\ $\bbeta=\hh'(\ZZ)$
    \item Output $A_2(\YY,\aalpha,\bbeta)$
    \end{enumerate}
  \end{minipage}
  \begin{minipage}{0.45\textwidth}
    {\bf Algorithm $A$}
    \begin{enumerate}
    \setlength\itemsep{.1em}
    \item $(\gg,\DD,\JJ)\gets \NO^*$
    \item $\YY\gets \DD^q$
    \item $\aalpha=\gg(\YY)$\\
    \item $\ZZ=A_1(\YY,\aalpha)$
    \item Let \\ $\bbeta=\gg(\ZZ)$
    \item Output $A_2(\YY,\aalpha,\bbeta)$
    \end{enumerate}
  \end{minipage}
\end{multicols}}}\\

We remind the reader that
$$\gd(x)=\left\{
\begin{array}{ll}
\ggamma(x)&  x\in \SS \\
\hh(x)&  x\not\in \SS, (\forall y\in \SS)\ x_\JJ\not=y_\JJ \mbox{\ or\ } d(x,y)>(0.5-\lambda)n\\
\ggamma(y)&  x\not\in \SS, (\exists y\in \SS)\ x_\JJ=y_\JJ \mbox{\ and\ } d(x,y)\le (0.5-\lambda)n\ \ \ .
\end{array}
\right.
$$

The first expectation in {\bf C3.3} is equal to the probability of
$$(\YY,\aalpha,\JJ)\mbox{\ is consistent and\ } A_2(\YY,\aalpha,\hh'(A_1(\YY,\aalpha)))=1,$$
where $(\gg,\DD,\JJ)\gets \NO^*$, $\YY\gets\DD^q$, $\aalpha=\gg(\YY)$, and $\hh'\gets \JUNTA_{\YY,\aalpha,\JJ}$.

The second expectation is the probability of
$$A_2(\YY,\aalpha,\gg(A_1(\YY,\aalpha)))=1,$$
where $(\gg,\DD,\JJ)\gets \NO^*$ and $\aalpha=\gg(\YY)$.

The distributions of $(\YY,\aalpha,\JJ,\DD)$ in the two cases are identical.

We say that a tuple $(Y,\alpha,J,\D)$ in the support of $(\YY,\aalpha,\JJ,\DD)$ is {\it good} if it satisfies the following three conditions: Here $Z=A_1(Y,\alpha)$ and $S$ is the support of $\D$
\begin{description}
\item[$E_0:$] $Y$ is scattered by $J$.
\item[$E_1:$] Every $z$ in $Z$ and every $x\in S\backslash\{y^{(i)}:i\in [q]\}$ have $d(x,z)> (0.5-\lambda)n$.
\item[$E_2:$] If a string $z$ in $Z$ satisfies $z_J=y_J$ for some $y$ in $Y$, then we have $d(y,z)\le (0.5-\lambda)n$.
\end{description}
We delay the proof of the following claim to the end.
\begin{claim}\label{claim2} We have that $(\YY,\aalpha,\JJ,\DD)$ is good with probability at least 7/8.\end{claim}
Fix any good $(Y,\alpha,J,\D)$ in the support and let $Z=A_1(Y,\alpha)$. We first show that since $Y$ is scattered by $J$ we have that $(Y,\alpha,J)$ is consistent. Let $i,j\in[q]$ with $y_J^{(i)}=y^{(i)}_J$. Since $Y$ is scattered by $J$ we have $i=j$ and therefore $\alpha_i=g(y^{(i)})=g(y^{(j)})=\alpha_j$. Therefore $(Y,\alpha,J)$ is consistent.

We finish the proof by showing
that the distribution of $\gg(Z)$, a binary string of length $q$, conditioning on $(\YY,\aalpha,\JJ,\DD) = (Y,\alpha,J,\D)$ is the same as that of $\hh'(Z)$ with $\hh'\gets \JUNTA_{Y,\alpha,J}$. This combined with Lemmas~\ref{L4E},~\ref{L5} and Claim~\ref{claim2} implies that the difference of the two probabilities has absolute value at most $1/8$.
To see this is the case, we partition strings of $Z$ into $Z_w$, where each $Z_w$ is a nonempty set that
contains all $z$ in $Z$ with $z_J = w \in  \{0, 1\}^J$ . For each $Z_w$, we consider the following two cases:
\begin{description}
\item[Case I.] There exists $y^{(i)}$ in $Y$ with $y^{(i)}_J=w$. By $E_0$, $y^{(i)}$ is the only string $y$ in $Y$ that satisfies $y_J=w$. By $E_2$, every $z\in Z_w$ satisfies $d(z,y^{(i)})\le (0.5-\lambda)n$. By $E_1$, every $z\in Z_w$ and every $y\in S\backslash \{y^{(i)}:i\in [q]\}$ we have $d(x,z)>(0.5-\lambda)n$. Therefore, the only $y$ in $S$ that satisfies $y_J=w$ and $d(z,y)\le (0.5-\lambda)n$ is $y^{(i)}$. Therefore, for every $z\in Z_w$ we have $\gg(z)=\gamma(y^{(i)})=\alpha_i$. On the other hand, for every $z\in Z_w$ and $\hh'\gets\JUNTA_{Y,\alpha,J}$ we have $\hh'(z)=\alpha_i$.
\item[Case II.] There exists no $y$ in $Y$ with $y_J=w$. By $E_1$ for every $z\in Z_w$ and every $x\in S\backslash \{y^{(i)}:i\in [q]\}$ we have that $d(x,z)\ge (0.5-\lambda)n$. Therefore for every $z\in Z_w$ and every $x\in S$ we have that $x_J\not=z_J$ or $d(x,z)\ge (0.5-\lambda)n$. Thus, for every $z\in Z_w$ we have that $g(z)=\hh(z) =\bb(w)$ for some uniform bit $\bb(w)$. The same is true for $\hh'\gets\JUNTA_{Y,\alpha,J}$.
\end{description}

So the conditional distribution of $\gg(Z)$ is identical to that of $\hh'(Z)$ with $\hh'\gets\JUNTA_{Y,\alpha,J}$.
This finishes the proof of {\bf C3}.

Now to prove Claim~\ref{claim2}, we show that $\Pr[\overline{E_0}]\le 1/16$ and $\Pr[\overline{E_1}],\Pr[\overline{E_2}]\le 1/32$. By the union bound we get
$$\Pr[E_1\mbox{\ and\ } E_2\mbox{\ and\ } E_3]\ge 1-\Pr[\overline{E_0}]-\Pr[\overline{E_1}]-\Pr[\overline{E_2}]\ge \frac{7}{8}.$$

\subsection{The Proof for $E_0$ and $E_1$}
From Claim~\ref{claim}, we have
$$\Pr[\overline{E_0}]\le \frac{1}{16}.$$

We now prove that with probability at most $1/32$,
\begin{description}
\item[$\overline{E_1}$:] There exists $z$ in $\ZZ$ and $x\in \SS\backslash\{\yy^{(i)}:i\in [q]\}$ such that $d(x,z)\le (0.5-\lambda)n$.
\end{description}

To prove that $\Pr[\overline{E_1}]\le {1}/{32}$, we fix a pair $(Y, \alpha)$ in the support and let $\ell\le q$ be the number of distinct strings in $Y$ and $Z = A_1 (Y, \alpha)$. Conditioning on $\YY = Y, \SS \backslash \YY$ is a uniformly random subset of $\{0, 1\}^n \backslash Y$ of size $m - \ell$. Instead of working with $\SS \backslash \YY$, we let $\TT$ denote a set obtained by making $m - \ell$ draws from $\{0, 1\}^n$ uniformly at random (with replacements).
On the one hand, by Lemma~\ref{L4E}, the total variation distance between $S \backslash Y$ and $\TT$ is exactly the probability that
either (1) $\TT \cap Y$ is nonempty or (2) $|\TT| < m - \ell$. By two union bounds, (1) happens with probability
$1-(1-\ell/2^n)^{m-\ell}\le (m - \ell) \cdot (\ell/2^n ) \le mq/2^n$ and, by Lemma~\ref{l2}, (2) happens with probability at most $m^2/2^{n}$. As a result, the total variation distance is at most $(mq +m^2)/2^n$. On the other hand, by Chernoff bound (\ref{Chernoff2}) in Lemma~\ref{Chernoff}, the probability that one of
the strings of $\TT$ has distance at most $(0.5-\lambda)n$ with one of the strings of $Z$ is at most $mq · \exp(-\lambda^2n)$. Thus, by union bound ($n\ge 15+2\log\log |C|$, $q\le 2^{k/2-3}$ and $m=18\ln |C|$)
$$\Pr[\overline{E_2}]\le \frac{mq +m^2}{2^n} +mq \cdot e^{-\lambda^2n} \le \frac{m^2}{2^{n-1}} +mq \cdot e^{-\lambda^2n}\le\frac{1}{64}+\frac{1}{64}\le\frac{1}{32}.$$

\subsection{The Proof for $E_2$}
We now prove that with probability at most $1/32$,
\begin{description}
\item[$\overline{E_2}$:] There exists two strings $z$ in $\ZZ$ and $y$ in $\YY$ that satisfies $z_\JJ=y_\JJ$ and $d(y,z)> (0.5-\lambda) n$.
\end{description}
Fix a pair $(Y, \alpha)$ in the support and let $Z = A_2 (Y, \alpha)$. Because $\JJ$ is independent from
$(\YY,\aalpha)$, it remains a subset of $[n]$ of size $k$ drawn uniformly at random. For each pair $(y, z)$ with $y$
from $Y$ and $z$ from $Z$ that satisfy $d(y, z) > (0.5-\lambda)n$, the probability of $y_J = z_J$ is at most
$$\frac{{(0.5+\lambda)n\choose k}}{{n\choose k}}\le (0.5+\lambda)^k.$$

Then
\begin{eqnarray*}
\Pr[\overline{E_2}]
&\le& q^2\cdot \frac{{(0.5+\lambda)n\choose k}}{{n\choose k}}\le q^2(0.5+\lambda)^k \le\frac{1}{32}.
\end{eqnarray*}

\bibliography{TestingRef}

\begin{thebibliography}{1}

\bibitem{Bshouty19}
Nader~H. Bshouty.
\newblock Almost optimal distribution-free junta testing.
\newblock In {\em 34th Computational Complexity Conference, {CCC} 2019, July
  18-20, 2019, New Brunswick, NJ, {USA}}, pages 2:1--2:13, 2019.
\newblock URL: \url{https://doi.org/10.4230/LIPIcs.CCC.2019.2}, \href
  {http://dx.doi.org/10.4230/LIPIcs.CCC.2019.2}
  {\path{doi:10.4230/LIPIcs.CCC.2019.2}}.

\bibitem{LiuCSSX18}
Zhengyang Liu, Xi~Chen, Rocco~A. Servedio, Ying Sheng, and Jinyu Xie.
\newblock Distribution-free junta testing.
\newblock In {\em Proceedings of the 50th Annual {ACM} {SIGACT} Symposium on
  Theory of Computing, {STOC} 2018, Los Angeles, CA, USA, June 25-29, 2018},
  pages 749--759, 2018.
\newblock URL: \url{https://doi.org/10.1145/3188745.3188842}, \href
  {http://dx.doi.org/10.1145/3188745.3188842}
  {\path{doi:10.1145/3188745.3188842}}.

\end{thebibliography}

\end{document}